\documentclass[a4paper,reqno]{amsart}
\numberwithin{equation}{section}
\usepackage{graphicx}
\usepackage{lscape}
\usepackage[utf8]{inputenc}
\begin{document}

\newcommand{\beqn}{\begin{equation}}
\newcommand{\eeqn}{\end{equation}}
\newcommand{\beqna}{\begin{eqnarray}}
\newcommand{\ee}{\end{eqnarray}}
\newcommand{\best}{\begin{eqnarray*}}
\newcommand{\ees}{\end{eqnarray*}}

\newcommand{\undertilde}[1]{\underset{\widetilde{}}{#1}}
\newcommand{\al}{\alpha}
\newcommand{\de}{\delta}
\newcommand{\De}{\Delta}
\newcommand{\be}{\beta}
\newcommand{\ga}{\gamma}
\newcommand{\Ga}{\Gamma}
\newcommand{\Be}{\Xi}
\newcommand{\Al}{\Lambda}

\newcommand{\ala}{\bar{\alpha}}
\newcommand{\alb}{\bar{\bar{\alpha}}}
\newcommand{\dea}{\bar{\delta}}
\newcommand{\bea}{\bar{\beta}}
\newcommand{\gaa}{\bar{\gamma}}

\newcommand{\als}{\hat{\alpha}}
\newcommand{\des}{\hat{\delta}}
\newcommand{\bes}{\hat{\beta}}
\newcommand{\gas}{\hat{\gamma}}

\newcommand{\alas}{\hat{\bar{\alpha}}}
\newcommand{\deas}{\hat{\bar{\delta}}}
\newcommand{\beas}{\hat{\bar{\beta}}}
\newcommand{\gaas}{\hat{\bar{\gamma}}}

\newcommand{\as}{\hat{a}}
\newcommand{\ds}{\hat{d}}
\newcommand{\bs}{\hat{b}}
\newcommand{\cs}{\hat{c}}

\newcommand{\ab}{\bar{a}}
\newcommand{\da}{\bar{d}}
\newcommand{\ba}{\bar{b}}
\newcommand{\ca}{\bar{c}}

\newcommand{\ass}{\hat{\hat{a}}}
\newcommand{\dss}{\hat{\hat{d}}}
\newcommand{\bss}{\hat{\hat{b}}}
\newcommand{\css}{\hat{\hat{c}}}

\newcommand{\aas}{\hat{\bar{a}}}
\newcommand{\das}{\hat{\bar{d}}}
\newcommand{\bas}{\hat{\bar{b}}}
\newcommand{\cas}{\hat{\bar{c}}}

\newcommand{\la}{\lambda}
\newcommand{\La}{\Lambda}
\newcommand{\laa}{\bar{\lambda}}
\newcommand{\lab}{\bar{\bar{\lambda}}}
\newcommand{\lac}{\bar{\bar{\bar{\lambda}}}}
\newcommand{\lad}{\overset{4}{\lambda}}
\newcommand{\lae}{\overset{5}{\lambda}}
\newcommand{\laf}{\overset{6}{\lambda}}

\newcommand{\mus}{\hat{\mu}}

\newcommand{\va}{\bar{v}}
\newcommand{\vb}{\bar{\bar{v}}}
\newcommand{\vs}{\hat{v}}
\newcommand{\vas}{\hat{\bar{v}}}

\newcommand{\ua}{\bar{u}}
\newcommand{\ub}{\bar{\bar{u}}}
\newcommand{\us}{\hat{u}}
\newcommand{\uas}{\hat{\bar{u}}}

\newcommand{\ta}{\bar{t}}
\newcommand{\tb}{\bar{\bar{t}}}
\newcommand{\ts}{\hat{t}}
\newcommand{\tas}{\hat{\bar{t}}}

\newcommand{\vo}{v_1}
\newcommand{\voa}{\bar{v}_1}
\newcommand{\vos}{\hat{v}_1}
\newcommand{\voas}{\hat{\bar{v}}_1}

\newcommand{\vt}{v_2}
\newcommand{\vta}{\bar{v}_2}
\newcommand{\vts}{\hat{v}_2}
\newcommand{\vtas}{\hat{\bar{v}}_2}

\newcommand{\vth}{v_3}
\newcommand{\vtha}{\bar{v}_3}
\newcommand{\vths}{\hat{v}_3}
\newcommand{\vthas}{\hat{\bar{v}}_3}

\newcommand{\vf}{v_4}
\newcommand{\vfa}{\bar{v}_4}
\newcommand{\vfs}{\hat{v}_4}
\newcommand{\vfas}{\hat{\bar{v}}_4}

\newcommand{\wa}{\bar{w}}
\newcommand{\ws}{\hat{w}}
\newcommand{\was}{\hat{\bar{w}}}
\newcommand{\wb}{\bar{\bar{w}}}
\newcommand{\wc}{\bar{\bar{\bar{w}}}}

\newcommand{\xa}{\bar{x}}
\newcommand{\xb}{\bar{\bar{x}}}
\newcommand{\xc}{\bar{\bar{\bar{x}}}}
\newcommand{\xd}{\overset{4}{x}}
\newcommand{\xe}{\overset{5}{x}}
\newcommand{\xf}{\overset{6}{x}}

\newcommand{\xs}{\hat{x}}
\newcommand{\xss}{\hat{\hat{x}}}
\newcommand{\xas}{\hat{\bar{x}}}
\newcommand{\xbs}{\hat{\bar{\bar{x}}}}

\newcommand{\ya}{\bar{y}}
\newcommand{\yb}{\bar{\bar{y}}}
\newcommand{\yc}{\bar{\bar{\bar{y}}}}
\newcommand{\yd}{\overset{4}{y}}
\newcommand{\ye}{\overset{5}{y}}
\newcommand{\yf}{\overset{6}{y}}

\newcommand{\ys}{\hat{y}}
\newcommand{\yss}{\hat{\hat{y}}}
\newcommand{\yas}{\hat{\bar{y}}}

\newcommand{\zs}{\hat{z}}
\newcommand{\za}{\bar{z}}
\newcommand{\zss}{\hat{\hat{z}}}
\newcommand{\zas}{\hat{\bar{z}}}

\newcommand{\ra}{\bar{r}}
\newcommand{\rb}{\bar{\bar{r}}}
\newcommand{\rc}{\bar{\bar{\bar{r}}}}
\newcommand{\rs}{\hat{r}}
\newcommand{\rss}{\hat{\hat{r}}}
\newcommand{\ras}{\hat{\bar{r}}}

\newcommand{\p}{\mathcal{p}}

\newcommand{\du}[3]{#1_{#2}^{#3}}
\newcommand{\ol}[1]{\overline{#1}}
\newcommand{\ul}[1]{\underline{#1}}
\newcommand{\os}[2]{\overset{#1}{#2}}

\newcommand\qon{{\rm qP}_{\rm{\scriptstyle I}}}
\newcommand\qtw{{\rm qP}_{\rm{\scriptstyle II}}}
\newcommand\qth{{\rm qP}_{\rm{\scriptstyle III}}}
\newcommand\qfi{{\rm qP}_{\rm{\scriptstyle V}}}
\newcommand\qsi{{\rm qP}_{\rm{\scriptstyle VI}}}

\newcommand\don{{\rm dP}_{\rm{\scriptstyle I}}}
\newcommand\dtw{{\rm dP}_{\rm{\scriptstyle II}}}
\newcommand\dth{{\rm dP}_{\rm{\scriptstyle III}}}

\newcommand\con{{\rm cP}_{\rm{\scriptstyle I}}}
\newcommand\ctw{{\rm cP}_{\rm{\scriptstyle II}}}

\newcommand{\eqn}[1]{(\ref{#1})}
\newcommand{\ig}[1]{\includegraphics{#1}}

\theoremstyle{break}    \newtheorem{Cor}{Corollary}
\theoremstyle{plain}    \newtheorem{Exa}{Example}[section]
\newtheorem{Rem}{Remark} \theoremstyle{marginbreak}
\newtheorem{Lem}[Cor]{Lemma}
\newtheorem{Def}[Cor]{Definition}
\newtheorem{prop}{Proposition}
\newtheorem{theorem}{Theorem}
\newtheorem{eg}{Example}
\newtheorem{fact}{Fact}

\title[LMKdV Hierarchies]{Lattice Modified KdV Hierarchy from a Lax pair expansion}
\author{Mike Hay}
\address{School of Mathematics and Statistics, F07 University of Sydney, NSW 2006, Australia}
\email{mhay@uni.sydney.edu.au}

\begin{abstract}
We produce a hierarchiy of integrable equations by systematically adding terms to the Lax pair for the lattice modified KdV equation. The equations in the hierarchy are related to one aonother by recursion relations. These recursion relations are solved explicitly so that every equation in the hierarchy along with its Lax pair is known.\end{abstract}
\maketitle

\section{Introduction}
The goal of this paper is to present an integrable hierarchy of equations associated with the lattice modified KdV equation (LMKdV). We do this by considering a Lax pair expansion which is polynomial in the spectral parameter. This kind of procedure is not new, in the same way hierarchies of integrable equations have been produced \cite{al75}-\cite{pz09}, including hierarchies of purely discrete equations starting with \cite{lm01}.  Where this work differs markedly from previously known results, is in that we can write formulae that express every equation in the hierarchy and its Lax pair explicitly.

In \cite{h09} two-component analogues of the LMKdV and the lattice sine-Gordon equation (LSG) were shown to represent the most general systems that can be associated with the Lax pairs for those equations. The two-component analogue of LMKdV is
\begin{subequations}\label{lmkdv2}
\beqna
\frac{\xas}{\xs}+\frac{\ya}{y}&=&\frac{\xas}{\xa}+\frac{\ys}{y},\\
\frac{\yas}{\ys}+\frac{\xa}{x}&=&\frac{\yas}{\ya}+\frac{\xs}{x},
\ee
\end{subequations}
where there are two discrete independent variables $l,m\in\mathbb{Z}$, the complex dependent variables are $x=x(l,m)$, $y=y(l,m)$ and shifts are represented by $\xa=x(l+1,m)$, $\xs=x(l,m+1)$ etc. This system was dubbed LMKdV$_2$ because setting $y=\la\mu/x$, where $\la(l)$ and $\mu(m)$ are arbitrary functions, retrieves the familiar LMKdV. A Lax pair for \eqn{lmkdv2} is
\beqn
L=\left(\begin{array}{cc}1&\nu y/\xa\\ \nu x/\ya&1\end{array}\right),\quad\quad M=\left(\begin{array}{cc}1&\nu y/\xs\\ \nu x/\ys&1\end{array}\right),\label{lmkdv2lp}
\eeqn
where $\nu$ is the spectral parameter and the Lax matrices, $L$ and $M$, appear in the linear system
\beqn
\bar{\theta}=L\theta, \quad\quad \hat{\theta}=M\theta,\nonumber
\eeqn
which leads to the compatibility condition
\beqn
\widehat{L}M=\ol{M}L.\label{ccMat}
\eeqn

In \cite{cjm06} an algorithm was developed to derive Lax pairs for the continuous modified KdV hierarchy, based on the AKNS expansion technique \cite{akns74}. In this paper we use a similar construction to form hierarchies of equations by adding terms to a generalisation of the $L$ matrix in \eqn{lmkdv2lp} as
\beqn
L_k=\left(\begin{array}{cc}a_0+a_1 \nu^2+\ldots+a_k \nu^{2k}&b_0 \nu+b_1\nu^3+\ldots+b_k \nu^{2k+1}\\c_0\nu+c_1\nu^3+\ldots +c_k\nu^{2k+1}&d_0+d_1\nu^2+\ldots+d_k\nu^{2k}\end{array}\right),\label{lk}
\eeqn
where $a_i(l,m)$,  $0\leq i\leq k$ are unknowns that will be found using the compatibility condition, other roman letters are similar. We keep the same generalised form of $M$ for all members of the hierarchy
\beqn
M=\left(\begin{array}{cc}1&\be \nu\\\ga\nu&1\end{array}\right),\label{m}
\eeqn
where the diagonal entries are set to unity using a gauge.

While the form of the Lax pairs is nearly identical to that used in \cite{cjm06}, the methodology and results of this article are quite different. We will derive a set of recursion relations in the variable $k$ from \eqn{lk}, which describe how the next Lax pair in the hierarchy can be found from the previous one. These recursion relations will be solved explicitly to give formulas for every Lax pair in the hierarchy as well as every equation. This will show that, as $k$ is increased, the Lax pairs become increasingly complex but the equations retain the same simple form throughout. However, the number of arbitrary non-autonomous terms that can be inserted in to the reduced equations increases at each level of the hierarchy.

Section \ref{sec:rr} will see the compatibility condition put into a convenient form, then use this to derive recursion relations that relate different members of the hierarchy. In Section \ref{sec:soln} we provide solutions to the recursion relations found in section \ref{sec:rr} and describe how to use these to construct the Lax pairs including two examples. In section \ref{sec:sys} we derive the associated nonlinear systems and show how these are reduced to scalar equations on stretched quad-graphs. The paper terminates with a discussion.

\section{Derivation of the recursion relations}\label{sec:rr}
In this section we derive the recursion relations that relate different members of the hierarchy. In section \ref{ssec:cc} we put the compatibility condition into the required form and in section \ref{ssec:rr} we derive the recursion relations.
\subsection{Compatibility condition}\label{ssec:cc}
Using a gauge to set $a_0=d_0=1$ and substituting \eqn{lk} and \eqn{m} into the compatibility condition \eqn{ccMat} gives one equation for each of the four entries of the 2$\times$2 system:
\beqna
(1,1)\quad\quad 1+\sum_{i=1}^k (\as_i+\ga\bs_{i-1})\nu^{2i}+\ga\bs_k\nu^{2k+2}&=&1+\sum_{i=1}^{k}(a_i+\bea c_{i-1})\nu^{2i}+\bea c_k\nu^{2k+2},\nonumber\\
(2,2)\quad\quad1+\sum_{i=1}^k (\ds_i+\be\cs_{i-1})\nu^{2i}+\be\cs_k\nu^{2k+2}&=&1+\sum_{i=1}^{k}(d_i+\gaa b_{i-1})\nu^{2i}+\gaa b_k\nu^{2k+2},\nonumber\\
(1,2)\quad\quad \sum_{i=0}^k(\be\as_i+\bs_i)\nu^{2i+1}&=&\sum_{i=0}^{k}(b_i+\bea d_i)\nu^{2i+1},\nonumber\\
(2,1)\quad\quad \sum_{i=0}^k(\ga\ds_i+\cs_i)\nu^{2i+1}&=&\sum_{i=0}^{k}(c_i+\gaa a_i)\nu^{2i+1}.\nonumber
\ee
As the spectral parameter $\nu$ is independent of the lattice variables $l$ and $m$, we separate the compatibility conditions by powers of $\nu$. This yields
\begin{subequations}\label{cc}
\beqna
\as_i+\ga \bs_{i-1}&=&a_i+\bea c_{i-1},\quad 0<i\leq k\label{cc11}\\
\ds_i+\be \cs_{i-1}&=&d_i+\gaa b_{i-1},\quad 0<i\leq k\label{cc22}\\
\be \as_i+\bs_i&=&\bea d_i+b_i,\quad 0\leq i\leq k\label{cc12}\\
\ga \ds_i+\cs_i&=&\gaa a_i+c_i,\quad 0\leq i\leq k\label{cc21}
\ee
\end{subequations}
In addition to these we must consider the two extreme powers of $\nu$ in the diagonal entries. The equations coming from diagonal entries of the compatibility condition at order $\nu^0$ are satisfied and at order $\nu^{2k+2}$ permit the following parametrisation:
\beqn
\begin{array}{ll}
b_k=y/\xa,&\quad c_k=x/\ya,\\
\be=y/\xs,&\quad \ga=x/\ys,
\end{array}\label{para}
\eeqn
where we have introduced the dependent variables $x(l,m)$ and $y(l,m)$, these will become the dependent variables in the nonlinear systems associated with the Lax pairs. 

Symmetry exists in the compatibility condition such that we may interchange
\beqn
\{[1,1],[1,2],[2,1],[2,2],a_i,b_i,c_i,d_i,\be,\ga\}\leftrightarrow \{[2,2],[2,1],[1,2],[1,1],d_i,c_i,b_i,a_i,\ga,\be\},\label{symmetry}
\eeqn
where numbers in square brackets indicate the matrix entries of the compatibility condition, [1,1], [2,2], [1,2] and [2,1] are given by \eqn{cc11}, \eqn{cc22}, \eqn{cc12} and \eqn{cc21}, respectively. Because everything in this paper is derived from the compatibility conditions, this symmetry naturally persists throughout.

It is possible to remove one of the shifts in $m$ from the compatibility conditions. This is done by starting at the $i$-th condition from \eqn{cc12}, removing $\bs_i$ using \eqn{cc11}, $\as_{i+1}$ by using the $i+1$-st in \eqn{cc12}, then successively removing the shifts in $m$ by using the same equations.
\beqna
\as_i&=&\frac{\bea}{\be}d_i+\frac{1}{\be}b_{i}-\frac{1}{\be}\bs_{i},\nonumber\\
&=&\frac{\bea}{\be}d_i+\frac{1}{\be}b_{i}-\frac{\bea}{\be\ga}c_i-\frac{1}{\be\ga}(-a_{i+1}+\frac{\bea}{\be}d_{i+1}+\frac{1}{\be}b_{i+1}-\frac{1}{\be}\bs_{i+1}),\nonumber
\ee
etc. Because the last relevant compatibility condition, at order $\nu^{2k+1}$, is $\ga \bs_k=\bea c_k$, the procedure stops when $i=k$ and we obtain
\beqn
\as_i=\frac{\bea}{\be}d_i+\frac{1}{\be}b_i-\frac{\bea}{\be\ga}c_i+\sum_{j=1}^{k-i}\frac{1}{(\be\ga)^j}\left(-a_{i+j}+\frac{b_{i+j}}{\be}-\frac{\bea c_{i+j}}{\be\ga}+\frac{\bea d_{i+j}}{\be}\right).\nonumber
\eeqn
An expression for $\bs_i$ is obtained in a similar way starting with \eqn{cc11}, we find
\beqn
\bs_i=\frac{\bea}{\ga}c_i+\sum_{j=1}^{k-i}\frac{1}{(\be\ga)^j}\left(\be a_{i+j}-b_{i+j}+\frac{\bea}{\ga} c_{i+j}-\bea d_{i+j}\right).\nonumber
\eeqn
Expressions for $c_i$ and $d_i$ are obtained by symmetry \eqn{symmetry}. 

We write the compatibility conditions in matrix form by introducing
\beqn
\psi_i=\left(\begin{array}{c}a_i\\d_i\end{array}\right),\quad\quad\varphi_i=\left(\begin{array}{c}b_i\\c_i\end{array}\right),\label{def_psiphi}
\eeqn
which means that the compatibility conditions are equivalent to
\begin{subequations}\label{ccm}
\beqna
\hat{\psi_i}&=&K\psi_i+H\varphi_i+\sum_{j=i+1}^k D^{(j)}_{i}\psi_{j}+E^{(j)}_{i}\varphi_j, \quad 1\leq i\leq k,\label{shang}\\
\psi_0&=&\left(\begin{array}{c}1\\1\end{array}\right),\label{shang0}\\
\hat{\varphi}_i&=&J\varphi_i+\sum_{j=i+1}^k F^{(j)}_{i}\psi_{j}+G^{(j)}_{i}\varphi_j, \quad 0\leq i\leq k,\label{xia}
\ee
\end{subequations}
where
\beqna
&K=\left(\begin{array}{cc}0&\bea/\be\\\gaa/\ga&0\end{array}\right),\quad H=\left(\begin{array}{cc}1/\be&-\bea/(\be\ga)\\-\gaa/(\be\ga)&1/\ga\end{array}\right),\quad J=\left(\begin{array}{cc}0&\bea/\ga\\\gaa/\be&0\end{array}\right),\nonumber\\
&D^{(j)}_{i}=1/(\be\ga)^{j-i}\left(\begin{array}{cc}-1&\bea/\be\\\gaa/\ga&-1\end{array}\right),\quad E^{(j)}_{i}=1/(\be\ga)^{j-i}\left(\begin{array}{cc}1/\be&-\bea/(\be\ga)\\-\gaa/(\be\ga)&1/\ga\end{array}\right),
\nonumber\\
&F^{(j)}_{i}=1/(\be\ga)^{j-i}\left(\begin{array}{cc}\be&-\bea\\-\gaa&\ga\end{array}\right),\quad\quad G^{(j)}_{i}=1/(\be\ga)^{j-i}\left(\begin{array}{cc}-1&\bea/\ga\\\gaa/\be&-1\end{array}\right).\nonumber
\ee

\subsection{Recursion relations} \label{ssec:rr}
To solve the compatibility condition \eqn{ccm} we perform the following sequence of calculations:
\begin{itemize}
\item begin with the known quantity: $\psi_0=(1\,\,1)^T$,
\item use \eqn{shang}$_{i=0}$ to find $\varphi_0$ in terms of $\psi_i$ and $\varphi_i$, $i\geq1$,
\item use \eqn{xia}$_{i=0}$ to find $\psi_1$ in terms of $\varphi_1$, $\psi_i$ and $\varphi_i$, $i\geq 2$,
\item repeat the las two steps at increasing values of $i$.
\end{itemize}
Performing these calculations shows that solutions must be of the following form.
\begin{subequations}\label{ansatz}
\beqna
\psi_i&=&A_i-P_i \varphi_i-\sum_{j=i+1}^k Q^{(j)}_i\psi_j+R^{(j)}_i\varphi_j,\label{ansatzpsi}\\
\varphi_i&=&B_i-\sum_{j=i+1}^k S^{(j)}_i\psi_j+T^{(j)}_i\varphi_j,\label{ansatzphi}
\ee
\end{subequations}
for $0<i\leq k$, where $A_i$, $B_i$, $P_i$, $Q^{(j)}_i$, $R^{(j)}_i$, $S^{(j)}_i$ and $T^{(j)}_i$ become the unknowns for which we will derive recurrence relations below.

Take $\widehat{\eqn{ansatzphi}}$ and use \eqn{ccm} to remove $\hat{\psi_i}$ and $\hat{\varphi_i}$ resulting in
\beqn
\begin{array}{l}J \varphi_i+\sum^k_{j=i+1}F^{(j)}_i\psi_j+G^{(j)}_i\varphi_j=\widehat{B}_i\\\quad\quad\quad\quad\quad\quad\quad-\sum_{j=i+1}^k\left[\widehat{S}^{(j)}_i\left(K\psi_j+H\varphi_j+\sum_{h=j+1}^kD^{(h)}_j\psi_h+E^{(h)}_j\varphi_h\right)\right]\\\quad\quad\quad\quad\quad\quad\quad\quad\quad\quad -\sum_{j=i+1}^k\left[\widehat{T}^{(j)}_i\left(J\varphi_j+\sum_{h=j+1}^kF^{(h)}_j\psi_h+G^{(h)}_j\varphi_h\right)\right].\end{array}\nonumber
\eeqn
Collect terms and note that
\beqn
\sum_{j=i+1}^k\sum_{h=j+1}^k\widehat{S}^{(j)}_iD^{(h)}_j\psi_h=\sum_{j=i+2}^k\sum_{h=i+1}^{j-1}\widehat{S}^{(h)}_iD^{(j)}_h\psi_j\nonumber
\eeqn
 to find
 \beqn
 \begin{array}{c}
 (\widehat{S}^{(i+1)}_{i}K-JS^{(i+1)}_{i}-F^{(i+1)}_{i})\psi_{i+1}+(\widehat{S}^{(i+1)}_{i}H+\widehat{T}^{(i+1)}_{i}J-JT^{(i+1)}_{i}+G^{(i+1)}_{i})\varphi_{i+1}=\quad\quad\quad\\
\widehat{B}_i-JB_i -\sum_{j=i+2}^k\left\{\left[\widehat{S}^{(j)}_iK-JS^{(j)}_i-F^{(j)}_i+\sum_{h=i+1}^{j-1}\left(\widehat{S}^{(h)}_iD^{(j)}_h+\widehat{T}^{(h)}_iF^{(j)}_h\right)\right]\psi_j\right.\\
\quad\quad\quad\quad
 +\left.\left[\widehat{S}^{(j)}_iH+\widehat{T}^{(j)}_i-JT^{(j)}_i-G^{(j)}_i+\sum_{h=i+1}^{j-1}\left(\widehat{S}^{(h)}_iE^{(j)}_h+\widehat{T}^{(h)}_iG^{(j)}_h\right)\right]\varphi_j\right\}.
 \end{array}
 \eeqn
 
Writing this in the same form as \eqn{ansatzpsi}$_{i+1}$ and introducing 
\[\Ga_i=\left(\widehat{S}^{(i+1)}_{i}K-JS^{(i+1)}_{i}+F^{(i+1)}_{i}\right)^{-1}
\nonumber \]
leads to the following set of relations
 \begin{subequations}\label{rr1}
 \beqna
 A_{i+1}&=&\Ga_i\left(\widehat{B}_i-JB_i\right),\\
 P_{i+1}&=&\Ga_i\left(\widehat{S}^{(i+1)}_{i}H+\widehat{T}^{(i+1)}_{i}J-JT^{(i+1)}_{i}+G^{(i+1)}_{i}\right),\\
 Q^{(j)}_{i+1}&=&\Ga_i\left[\widehat{S}^{(j)}_iK-JS^{(j)}_i+F^{(j)}_i+\sum_{h=i+1}^{j-1}\left(\widehat{S}^{(h)}_iD^{(j)}_h+\widehat{T}^{(h)}_iF^{(j)}_h\right)\right],\label{rr1r}\\
 R^{(j)}_{i+1}&=&\Ga_i\left[\widehat{S}^{(j)}_iH+\widehat{T}^{(j)}_i-JT^{(j)}_i+G^{(j)}_i+\sum_{h=i+1}^{j-1}\left(\widehat{S}^{(h)}_iE^{(j)}_h+\widehat{T}^{(h)}_iG^{(j)}_h\right)\right].\quad\quad\quad
 \ee
 \end{subequations}
 The same argument can be followed starting with $\widehat{\eqn{ansatzpsi}}$ which results in the following set of relations
 \begin{subequations}\label{rr2}
 \beqna
 B_i&=&\La_i\left(\widehat{A}_i-KA_i\right),\\
 S^{(j)}_i&=&\La_i\left[\widehat{P}_iF^{(j)}_i+\widehat{Q}^{(j)}_iK-KQ^{(j)}_i+D^{(j)}_i+\sum_{h=i+1}^{j-1}\left(\widehat{Q}^{(h)}_iD^{(j)}_h+\widehat{R}^{(h)}_iF^{(j)}_h\right)\right],\qquad\,\,\,\\
  T^{(j)}_i&=&\La_i\left[\widehat{Q}^{(j)}_iH+\widehat{R}^{(j)}_iJ-KR^{(j)}_i+\widehat{P}_iG^{(j)}_i+E^{(j)}_i+\right.\nonumber\\
  &&\left.\qquad\qquad\qquad\qquad\qquad\qquad\sum_{h=i+1}^{j-1}\left(\widehat{Q}^{(h)}_iE^{(j)}_h+\widehat{R}^{(h)}_iG^{(j)}_h\right)\right],
  \ee
  \end{subequations}
where $\La_i=\left(\widehat{P}_iJ-KP_i+H\right)^{-1}$. Equations \eqn{rr1} and \eqn{rr2} together form a set of nonlinear recursion relations that describe any Lax pair of the form \eqn{lk} with an arbitrary number, $k$, of terms in each entry of the $L$ matrix. If a solution to these recursion relations can be found, then $\psi_i$ and $\varphi_i$ can be reconstructed from \eqn{ansatz} which is a set of linear algebraic equations. In the next section solutions to the recursion relations are given.

\section{Solution to the recursion relations}\label{sec:soln}
In this section we present exact solutions to the recursion relations derived above. These solutions enable us to construct every equation and Lax pair in the hierarchy explicitly (see section \ref{sec:sys}). We use the solutions to construct the Lax pairs for the second and third members of the hierarchy in section \ref{ssec:recon}. Some important alternative forms of the solutions are given in section \ref{ssec:alt}

\subsection{The solutions}
We introduce the following notation:
\beqn
\os{h}{a}=a(l,m+h),\nonumber
\eeqn
In particular $\os{h}{(\be\ga)}=\be(l,m+h)\ga(l,m+h)$.

Close inspection of \eqn{rr1} and \eqn{rr2} suggests the following theorem.
\begin{theorem}\label{thm:soln}
The set of recursion relations defined by \eqn{rr1} and \eqn{rr2} are satisfied by the following formulae.
\begin{subequations}\label{solnac}
\beqna
  A_i&=&\left(\begin{array}{c}\prod_{j=1}^i\os{2j-2}{\ga}\,\os{2j-1}{\be}\\\prod_{j=1}^i\os{2j-2}{\be}\,\os{2j-1}{\ga}\end{array}\right),\label{solna}\\
  B_i&=&\left(\begin{array}{c}\be\prod_{j=1}^i\os{2j-1}{\ga}\,\os{2j}{\be}\\\ga\prod_{j=1}^i\os{2j-1}{\be}\,\os{2j}{\ga}\end{array}\right).\label{solnc}
\ee
\end{subequations}

Let $P_i=\left(\begin{array}{cc}P_{1,i}&0\\0&P_{2,i}\end{array}\right)$ we have
\beqn
P_{1,i}=-\sum_{j=1}^i\left[\left(\frac{\os{2j-2}{\ga}}{\os{2j-1}{\ga}}+1\right)\frac{1}{\be}\prod_{h=1}^{j-1}\frac{\os{2h-2}{\ga}\,\os{2h-1}{\be}}{\os{2h-1}{\ga}\,\os{2h}{\be}}\right],\label{solnb}
\eeqn
and $P_{2,i}$ is the same except $\be$ and $\ga$ are interchanged.

We express the quantity $\Ga_i$ as 
\beqn
\Ga_i=\frac{1}{1-\bea\gaa/\os{2i+1}{(\be\ga)}}\left(\begin{array}{cc}\ga&\Ga_{1,i}\\\Ga_{2,i}&\be\end{array}\right)
\eeqn
 with
\beqn
\Ga_{1,i}=\frac{\bea\ga}{\gas}\prod_{j=1}^i\frac{\os{2j-1}{\be}\,\os{2j}{\ga}}{\os{2j}{\be}\,\os{2j+1}{\ga}}\label{solnga}
\eeqn
and $\Ga_{2,i}$ is obtained from $\Ga_{1,i}$ by interchanging $\be$ and $\ga$. 

Similarly $\La_i=\left(\begin{array}{cc}\be&\La_{1,i}\\\La_{2,i}&\ga\end{array}\right)$ with
\beqn
\La_{1,i}=\bea\prod_{j=1}^i\frac{\os{2j-2}{\be}\,\os{2j-1}{\ga}}{\os{2j-1}{\be}\,\os{2j}{\ga}}\label{solde}
\eeqn

The matrices $Q^{(j)}_i$, $R^{(j)}_i$, $S^{(j)}_i$ and $T^{(j)}_i$ are all diagonal and the $[2,2]$ elements are obtained by interchanging $\be$ and $\ga$ in the corresponding $[1,1]$ elements. As such, we write only the $[1,1]$ elements of these solutions below. The $[1,1]$ element of $S^{(j)}_i$ is expressed as
\beqn
S^{(j)}_{i_{[1,1]}}=\sum_{h=0}^{2i}\frac{s^{(h)}_i}{(\os{h}{\be\ga})^j},\label{solnp}
\eeqn 
where
\beqna
{s}^{(h)}_i&=&-U_i^{(h)}(\os{h}{\be\ga})\sum_{g=0}^i \xi^{(2g)}_{2i}(\ga)(\os{h}{\be\ga})^g,\nonumber\\
U_i^{(h)}&=&(\os{h}{\be\ga})^{i-1}\frac{\be\prod_{g=1}^i\os{2g-1}{\ga}\,\os{2g}{\be}}{\prod_{g=0,\,g\neq h}^{2i}\left[(\os{g}{\be\ga})-(\os{h}{\be\ga})\right]},\nonumber\\
\xi^{(2g)}_{2i}(\ga)&=&\sum_{h_1=0}^{2i}\sum_{h_2=0}^{h_1-1}\ldots\sum_{h_{2(i-g)}=0}^{h_{2(i-g)-1}-1}\os{h_1}{\ga}\,\os{h_2}{\be}\ldots\os{h_{2(i-g)-1}}{\ga}\,\os{h_{2(i-g)}}{\be}\label{xi}
\ee
(See notes on the definition of $\xi$ after the end of this theorem.)

The $[1,1]$ element of $T^{(j)}_i$ is expressed as
\beqn
T^{(j)}_{i_{[1,1]}}=\sum_{h=0}^{2i}\frac{t^{(h)}_i}{(\os{h}{\be\ga})^j},\label{solnq}
\eeqn 
where
\beqn
{t}^{(h)}_i=U_i^{(h)}\sum_{g=0}^i\eta_{2i}^{(2g)}(\ga)(\os{h}{\be\ga})^g\nonumber
\eeqn
and
\beqn
\eta_{2i}^{(2g)}(\ga)=\sum_{h_1=0}^{2i}\sum_{h_2=0}^{h_1-1}\ldots\sum_{h_{2(i-g)}=0}^{h_{2(i-g)-1}-1}\sum_{h_{2(i-g)+1}=0}^{h_{2(i-g)}-1}\os{h_1}{\ga}\,\os{h_2}{\be}\ldots\os{h_{2(i-g)}}{\be}\,\os{h_{2(i-g)+1}}{\ga}.\label{eta}
\eeqn

The $[1,1]$ element of $Q^{(j)}_i$ is expressed as
\beqn
Q^{(j)}_{i_{[1,1]}}=\sum_{h=0}^{2i-1}\frac{q^{(h)}_i}{(\os{h}{\be\ga})^j},\label{solnr}
\eeqn 
where
\beqna
{q}^{(h)}_i&=&-V_i^{(h)}\sum_{g=0}^{i}\eta^{(2g)}_{2i-1}(\ga)(\os{h}{\be\ga})^g,\nonumber\\
\eta^{(2g)}_{2i-1}(\ga)&=&\sum_{h_1=0}^{2i-1}\sum_{h_2=0}^{h_1-1}\ldots\sum_{h_{2(i-g)-1}=0}^{h_{2(i-g)-2}-1}\sum_{h_{2(i-g)}=0}^{h_{2(i-g)-1}-1}\os{h_1}{\ga}\,\os{h_2}{\be}\ldots\os{h_{2(i-g)-1}}{\ga}\,\os{h_{2(i-g)}}{\be},\nonumber\\
V_i^{(h)}&=&(\os{h}{\be\ga})^{i-1}\frac{\prod_{g=1}^i\os{2g-2}{\ga}\,\os{2g-1}{\be}}{\prod_{g=0,\,g\neq h}^{2i-1}\left[(\os{g}{\be\ga})-(\os{h}{\be\ga})\right]}.\nonumber
\ee

The $[1,1]$ element of $R^{(j)}_i$ is expressed as
\beqn
R^{(j)}_{i_{[1,1]}}=\sum_{h=0}^{2i-1}\frac{r^{(h)}_i}{(\os{h}{\be\ga})^j},\label{solns}
\eeqn 
where
\beqna
{r}^{(h)}_i&=&V_i^{(h)}\sum_{g=0}^{i}\xi^{(2g)}_{2i-1}(\ga)(\os{h}{\be\ga})^g\nonumber\\
\xi^{(2g)}_{2i-1}(\ga)&=&\sum_{h_1=0}^{2i-1}\sum_{h_2=0}^{h_1-1}\ldots\sum_{h_{2(i-g)-2}=0}^{h_{2(i-g)-3}-1}\sum_{h_{2(i-g)-1}=0}^{h_{2(i-g)-2}-1}\os{h_1}{\ga}\,\os{h_2}{\be}\ldots\os{h_{2(i-g)-2}}{\be}\,\os{h_{2(i-g)-1}}{\ga}.\nonumber
\ee
\end{theorem}
\begin{proof}
The proof simply involves substituting the solutions into \eqn{rr1} and \eqn{rr2} to show that they are satisfied. However, it is rather technical and so it is left to Appendix \ref{app:solnpf}.
\end{proof}
Some notes on the definitions of $\xi^{(2g)}_{2i}(\ga)$: the upper index affects both the upper limit of the first sum and the number of factors (sums), the argument of the function (in this case $\ga$) refers to the first factor in the sum, factors alternate between $\be$ and $\ga$ after that. The definition of $\eta^{(2g)}_{2i}(\ga)$ is similar except there is one extra factor (and one extra sum). In general
\beqna
\xi^{(b)}_a(\gamma)&=&\sum_{h_1=0}^{a}\sum_{h_2=0}^{h_1-1}\ldots\sum_{h_{a-b-1}=0}^{h_{a-b-2}-1}\sum_{h_{a-b}=0}^{h_{a-b-1}-1}\os{h_1}{\ga}\,\os{h_2}{\be}\ldots\os{h_{a-b-1}}{\be}\,\os{h_{a-b}}{\ga},\nonumber\\
\eta^{(b)}_a(\gamma)&=&\sum_{h_1=0}^{a}\sum_{h_2=0}^{h_1-1}\ldots\sum_{h_{a-b}=0}^{h_{a-b-a}-1}\sum_{h_{a-b+1}=0}^{h_{a-b}-1}\os{h_1}{\ga}\,\os{h_2}{\be}\ldots\os{h_{a-b}}{\ga}\,\os{h_{a-b+1}}{\be},\nonumber
\ee
where $a-b>0$ is assumed to be odd in the above formulae, but the only difference when $a-b$ is even is that $\be$ and $\ga$ must be swapped in the last two factors. Although the lower limits of all the sums are zero, we adopt the convention that $\sum_{h=c}^d f(h)=0$ when $d<c$, so the effective lower limit is greater than zero for many of the sums. For $a=b$ we let $\xi_a^{(a)}=1$, the case where $a<b$ does not arise.

\subsection{Construction of the Lax pairs}\label{ssec:recon}
Here we explain how the solutions given in the previous section are used to construct the Lax pair for any member of the hierarchy. 

For $k=0$ the Lax pair is already known, and is written immediately from \eqn{para}. Starting with the formal solutions \eqn{ansatz}, it is obvious that the following linear system must hold when $k>0$
\beqn
Z_k\Psi_k=C_k,\label{recon}
\eeqn
where

\beqn
Z_k=
\left[\begin{array}{ccccccccc}
I&S_0^{(1)}&T_0^{(1)}&S_0^{(2)}&T_0^{(2)}&\ldots&S_0^{(k-1)}&T_0^{(k-1)}&S_0^{(k)}\\
0&I&P_1&Q_1^{(2)}&R_1^{(2)}&\ldots&Q_1^{(k-1)}&R_1^{(k-1)}&Q_1^{(k)}\\
0&0&I&S_1^{(2)}&T_1^{(3)}&\ldots&S_1^{(k-1)}&T_1^{(k-1)}&S_1^{(k)}\\
0&0&0&I&P_2&\ldots&Q_2^{(k-1)}&R_2^{(k-1)}&Q_2^{(k)}\\
\vdots& & & &\vdots&\ddots&\vdots&  &\vdots\\
0&\ldots& & & 0&\ldots&I&P_{k-1}&Q_{k-1}^{(k)}\\
\vdots&\ddots&&&\vdots&\ddots&0&I&S_{k-1}^{(k)}\\
0&&&&0&\ldots&0&0&I
\end{array}\right],\nonumber
\eeqn
\beqn
I=\left(\begin{array}{cc}1&0\\0&1\end{array}\right),\nonumber
\eeqn
\beqn
\Psi_k=
\left[\begin{array}{ccccccc}
\varphi_0&\psi_1&\varphi_1&\ldots&\psi_{k-1}&\varphi_{k-1}&\psi_k
\end{array}\right]^T\nonumber
\eeqn
and
\beqn
C_k
=\left[\begin{array}{c}
B_0-T_0^{(k)}\varphi_k\\
A_1-R_1^{(k)}\varphi_k\\
B_1-T_1^{(k)}\varphi_k\\
\vdots\\
A_{k-1}-R_{k-1}^{(k)}\varphi_k\\
B_{k-1}-T_{k-1}^{(k)}\varphi_k\\
A_k-P_k\varphi_k
\end{array}\right].\nonumber
\eeqn
Both $U_k$ and $C_k$ are known from the formulae in the previous section so this system can be solved for $\psi_i$ and $\varphi_i$, $1\leq i< k$ ($\psi_0$ and $\varphi_k$ are also known from \eqn{para}). The Lax pair for the $k$-th member of the hierarchy is then obtained from \eqn{def_psiphi} and \eqn{lk}.

As an example, consider $k=1$ where we have
\beqn
\Psi_1=\left(\begin{array}{c}\varphi_0\\\psi_1\end{array}\right), 
\quad Z_1=\left(\begin{array}{cc}I&S_0^{(1)}\\0&I\end{array}\right), 
\quad C_1=\left(\begin{array}{c}B_0-T_0^{(1)}\varphi_1\\A_1-P_1\varphi_1\end{array}\right).\nonumber
\eeqn
Using \eqn{solnac}, \eqn{solnb}, \eqn{solnp} and \eqn{solnq} we find
\begin{align}
S_0^{(1)}&=\left(\begin{array}{cc}-1/\ga&0\\0&-1/\be\end{array}\right),
&A_1&=\left(\begin{array}{c}\ga\bes\\\be\gas\end{array}\right),\nonumber\\
T_0^{(1)}&=\frac{1}{\be\ga}\left(\begin{array}{cc}1&0\\0&1\end{array}\right),
&B_0&=\left(\begin{array}{c}\be\\\ga\end{array}\right)\nonumber
\end{align}
and
\beqn
P_1=\left(\begin{array}{cc}
(\ga/\gas+1)/\be&0\\
0&(\be/\bes+1)/\ga
\end{array}\right).\nonumber
\eeqn
Solving \eqn{recon} and noting \eqn{para} gives us the terms in the $L$ matrix from the Lax pair, which is
\beqn
L=\left(\begin{array}{cc} 1+\nu^2(\xs/\xa+x\yss/(\xa\ys)+x/\xss)&
\nu(y/\xs+\ys/\xss+\yss/\xa)+\nu^3y/\xa\\
\nu(x/\ys+\xs/\yss+\xss/\ya)+\nu^3x/\ya&
1+\nu^2(\ys/\ya+y\xss/(\ya\xs)+y/\yss)
\end{array}\right).\nonumber
\eeqn
We also include the $L$ matrix from the Lax pair when $k=2$:
\begin{align}
a_0&=1,\nonumber\\
a_1&=\frac{x}{\xss}+\frac{\xs}{\os{3}{x}}+\frac{\xss}{\os{4}{x}}+\frac{\os{3}{x}}{\xa}+\frac{x\yss}{\os{3}{x}\ys}+\frac{\xs\os{3}{y}}{\os{4}{x}\yss}+\frac{\xss\os{4}{y}}{\xa\os{3}{y}}+\frac{x\os{3}{y}}{\os{4}{x}\ys}+\frac{\xs\os{4}{y}}{\xa\yss}+\frac{x\os{4}{y}}{\xa\ys},\nonumber\\
a_2&=\frac{x}{\os{4}{x}}+\frac{\xs}{\xa}+\frac{x\yss}{\xa\ys}+\frac{x\os{4}{y}}{\xa\os{3}{y}}+\frac{x\os{3}{x}}{\xa\xss},\nonumber\\
b_0&=\frac{y}{\xa}+\frac{\ys}{\xss}+\frac{\yss}{\os{3}{x}}+\frac{\os{3}{y}}{\os{4}{x}}+\frac{\os{4}{y}}{\xa},\nonumber\\
b_1&=\frac{y}{\os{3}{x}}+\frac{\ys}{\os{4}{x}}+\frac{\yss}{\xa}+\frac{\xss y}{\xs\os{4}{x}}+\frac{\os{3}{x}\ys}{\xss\xa}+\frac{y\os{3}{y}}{\os{4}{x}\yss}+\frac{\ys\os{4}{y}}{\xa\os{3}{y}}+\frac{\os{3}{x}y}{\xa\xs}+\frac{y\os{4}{y}}{\xa\yss}+\frac{\xss y\os{4}{y}}{\xa\xs\os{3}{y}},\nonumber\\
b_2&=\frac{y}{\xa}.\nonumber
\end{align}
The other coefficients in the Lax pair, $c_i$ and $d_i$, $i=0,1,2$, are obtained from the corresponding $b_i$ and $a_i$, respectively, by interchanging $x$ and $y$. We then form $L$ using \eqn{lk}.

\subsection{Alternative forms of the solutions}\label{ssec:alt}
The main building blocks in the solutions are sums such as
\beqn
\sum_{j=0}^{i}\eta^{(2j)}_{2i}(\ga)(\os{g}{\be\ga})^j=\sum_{j=0}^i\os{g}{(\be\ga)}{}^j\sum_{h_1=0}^{2i}\sum_{h_2=0}^{h_1-1}\ldots\sum_{h_{2(i-j)+1}=0}^{h_{2(i-j)}-1}\os{h_1}{\ga}\,\os{h_2}{\be}\ldots\os{h_{2(i-j)+1}}{\be}.\nonumber
\eeqn
In this section, we survey some properties of these sums.

We make the claim
\beqn
\sum_{j=0}^{i}\eta^{(2j)}_{2i}(\ga)(\os{g}{\be\ga})^j=\sum_{h_1=0}^{2i}\sum_{h_2=0}^{h_1-1+\delta_g^{h_1}}\ldots\sum_{h_{2i+1}=0}^{h_{2i}-1+\delta_{g}^{h_{2i}}}\os{h_1}{\ga}\os{h_2}{\be}\ldots\os{h_{2i+1}}{\ga},\label{sumdelta}
\eeqn
where $\delta_g^h$ is the kroenecker $\delta$. Evidence to support this claim follows. It can be helpful to display all of the terms in such sums diagrammatically, an example of which with $i=4$, $g=5$ is given in figure \ref{fig:sumDiagram},  where each non-increasing path from a node on the left of this figure, to one on the right represents a term in the sum.
\begin{figure}[h!]
\includegraphics{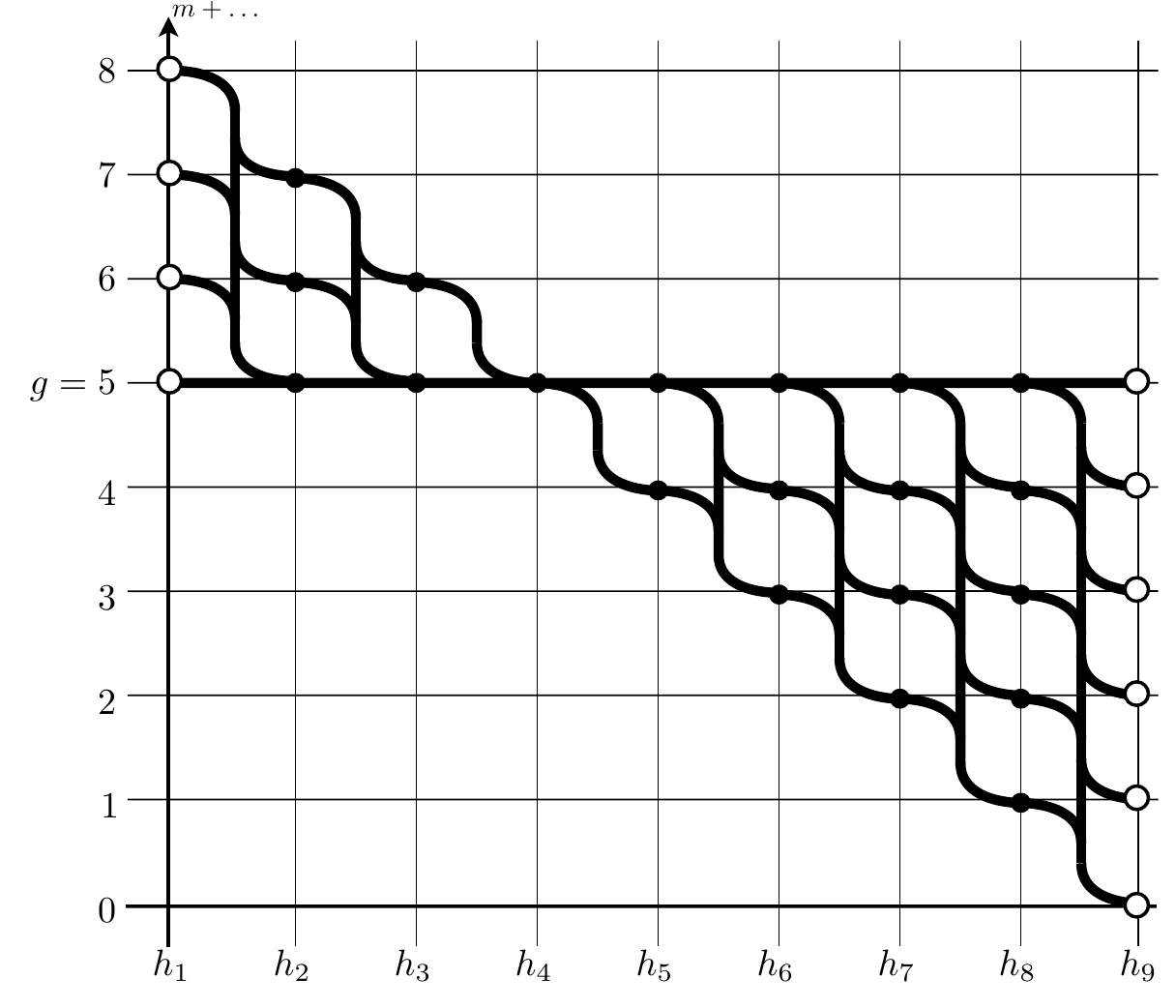}
\caption{Diagrammatic representation of sums such as those in \eqn{sumdelta}. Any non-increasing path from a node on the left to one on the right represents a term in the sum.}
\label{fig:sumDiagram}
\end{figure}

It is possible to separate the sums by removing terms containing the highest shift in $m$ as follows
\beqna
\sum_{h_1=0}^{2i}\ldots\sum_{h_{2i}=0}^{h_{2i-1}-1+\delta_{g}^{h_{2i-1}}}\os{h_1}{\be}\ldots\os{h_{2i}}{\ga} &=& \os{2i}{\be}\sum_{h_1=0}^{2i-1}\ldots\sum_{h_{2i-1}=0}^{h_{2i-2}-1+\delta_{g}^{h_{2i-2}}}\os{h_1}{\ga}\ldots\os{h_{2i-1}}{\ga}+\nonumber\\
&&\sum_{h_1=0}^{2i-1}\ldots\sum_{h_{2i}=0}^{h_{2i-1}-1+\delta_{g}^{h_{2i-1}}}\os{h_1}{\be}\ldots\os{h_{2i}}{\ga},\label{qitop}
\ee
or the lowest shift in $m$ as follows
\beqna
\sum_{h_1=0}^{2i}\ldots\sum_{h_{2i}=0}^{h_{2i-1}-1+\delta_{g}^{h_{2i-1}}}\os{h_1}{\be}\ldots\os{h_{2i}}{\ga}&=&\ga\sum_{h_1=1}^{2i}\ldots\sum_{h_{2i-1}=1}^{h_{2i-2}-1+\delta_{g}^{h_{2i-2}}}\os{h_1}{\ga}\ldots\os{h_{2i-1}}{\be}+\nonumber\\
&&\sum_{h_1=1}^{2i}\ldots\sum_{h_{2i}=1}^{h_{2i-1}-1+\delta_{g}^{h_{2i-1}}}\os{h_1}{\be}\ldots\os{h_{2i}}{\ga}.\label{qibot}
\ee
Such decompositions are often useful and they appear again in appendix \ref{app:solnpf}. This process can be repeated and, when the number of factors reduces to two less than the difference between the limits of the first sum (the number of possible $m$-shifts), a factor of $(\os{g}{\be\ga})$ can be taken outside the sum
\beqn
\sum_{h_1=1}^{2i}\ldots\sum_{h_{2i+1}=1}^{h_{2i}-1+\delta_{g}^{h_{2i}}}\os{h_1}{\be}\ldots\os{h_{2i+1}}{\be} = (\os{g}{\be\ga})\sum_{h_1=1}^{2i}\ldots\sum_{h_{2i-1}=1}^{h_{2i-2}-1+\delta_{g}^{h_{2i-2}}}\os{h_1}{\be}\ldots\os{h_{2i-1}}{\be}.\label{qibg}
\eeqn

To see that \eqn{sumdelta} holds, we proceed in this way until all of the repeated iterates have been removed, producing a polynomial in $\os{g}{(\be\ga)}$.  Combining the coefficients of this polynomial to reproduce \eqn{sumdelta} is not a trivial exercise, but we omit it here because it is not essential to other parts of this article.

A similar argument justifies the following expression for the sums involving $\xi$:
\beqn
\sum_{h_1=0}^{2i}\ldots\sum_{h_{2i}=0}^{h_{2i-1}-1+\de_g^{h_{2i-1}}}\os{h_1}{\be}\,\os{h_2}{\ga}\ldots\os{h_{2i-1}}{\be}\,\os{h_{2i}}{\ga}=\sum_{j=0}^i \xi^{(2j)}_{2i}(\beta)(\os{g}{\be\ga})^j.\label{xi1}
\eeqn

Equations \eqn{sumdelta} to \eqn{xi1} also lead to another alternative expression of the sums, this time as $2\times 2$ matrix products, which can be useful for computer algebra calculations. Define
\beqna
\omega_{a,b}^{(1)}(\be;g)&=&\sum_{h_1=b}^a\ldots\sum_{h_{a-b+1}=b}^{h_{a-b}-1+\delta_g^{h_{a-b}}}\os{h_1}\be\os{h_2}{\ga}\ldots\os{h_{a-b+1}}{\beta},\label{w1b}\\
\omega_{a,b}^{(0)}(\be;g)&=&\sum_{h_1=b}^a\ldots\sum_{h_{a-b}=b}^{h_{a-b-1}-1+\delta_g^{h_{a-b-1}}}\os{h_1}\be\os{h_2}{\ga}\ldots\os{h_{a-b}}{\gamma},\label{w0b}
\ee
where $a-b$ is assumed to be even, when $a-b$ is odd the last factors of each sum swap from $\be$ to $\gamma$ or {\em vice versa}. If the argument of $\omega_{a,b}^{(i)}(\be;g)$, $i=0,1$, is changed from $\be$ to $\ga$, then the first factor in the sums changes from $\be$ to $\ga$ and the remaining factors alternate as usual. There are $a-b$ factors (sums) in terms in $\omega^{(0)}_{a,b}$ while $\omega_{a,b}^{(1)}(\be;g)$ has $a-b+1$ factors. When $a=2i$ and $b=0$, \eqn{w1b} and \eqn{w0b} become the same quantities as those in \eqn{sumdelta} and \eqn{xi1} respectively.

Using the same idea behind \eqn{qitop} and \eqn{qibg}, we find that
\beqn
\left(\begin{array}{c}\omega_{a,b}^{(1)}(\be;g)\\
\omega_{a,b}^{(0)}(\ga;g) \end{array}\right)=
\left(\begin{array}{cc}\os{a}{\be}&\os{g}{(\be\ga)}\\1&\os{a}{\ga}\end{array}\right)
\left(\begin{array}{c}\omega_{a-1,b}^{(1)}(\ga;g)\\
\omega_{a-1,b}^{(0)}(\be;g) \end{array}\right),\label{front}
\eeqn
where $a>g$ and $b\leq g$. Using the same idea behind \eqn{qibot} and \eqn{qibg} yields
\beqn
\left(\begin{array}{c}\omega_{a,b}^{(1)}(\be;g)\\
\omega_{a,b}^{(0)}(\be;g) \end{array}\right)=
\left(\begin{array}{cc}\os{b}{\be}&\os{g}{(\be\ga)}\\1&\os{b}{\ga}\end{array}\right)
\left(\begin{array}{c}\omega_{a,b+1}^{(1)}(\be;g)\\
\omega_{a,b+1}^{(0)}(\be;g) \end{array}\right),\label{back}
\eeqn
where $a\geq g$, $b<g$ and $a-b$ is even. By repeatedly applying \eqn{front} and \eqn{back} we can remove the sums in these quantities and replace them with $2\times 2$ matrix products. Noting that 
\beqn
\left(\begin{array}{c}\omega_{g,g}^{(1)}(\be;g)\\
\omega_{g,g}^{(0)}(\ga;g) \end{array}\right)=\left(\begin{array}{c}\os{g}{\be}\\
1\end{array}\right),
\nonumber \eeqn
we find that
\beqna
\left(\begin{array}{c}\omega_{a,b}^{(1)}(\al;g)\\
\omega_{a,b}^{(0)}(\al;g) \end{array}\right)&=&\prod_{h=b}^{g-1}W_h^{(g)}(\be)\cdot\left[\left(\begin{array}{cc}1&0\\0&0\end{array}\right)\widetilde{\prod}_{h=g+1}^{a}W_h^{(g)}(\be)\left(\begin{array}{c}\os{g}{\be}\\1\end{array}\right)\right.\nonumber\\
&&\quad\quad\left.+\left(\begin{array}{cc}0&0\\0&1\end{array}\right)\widetilde{\prod}_{h=g+1}^{a}W_h^{(g)}(\ga)\left(\begin{array}{c}\os{g}{\ga}\\1\end{array}\right)\right], \label{product}
\ee
where products with a tilde, $\widetilde{\prod}$, are to be written in reverse order and
\beqn
W_h^{(g)}(\be) = \begin{cases} \left(\begin{array}{cc}\os{h}{\be}&\os{g}{(\be\ga)}\\1&\os{h}{\ga}\end{array}\right), & \mbox{if } h-g\mbox{ is even}, \\ \left(\begin{array}{cc}\os{h}{\ga}&\os{g}{(\be\ga)}\\1&\os{h}{\be}\end{array}\right), & \mbox{if } h-g \mbox{ is odd}, \end{cases}\nonumber
\eeqn
\beqn
\al = \begin{cases}\be, & \mbox{if } h-g\mbox{ is even}, \\ \ga, & \mbox{if } h-g \mbox{ is odd}. \end{cases}\nonumber
\eeqn
One can construct any of the relevant sums as matrix products by using \eqn{product}.

\section{Associated nonlinear systems}\label{sec:sys}
In this section we derive the nonlinear systems associated with the hierarchy of Lax pairs that are described by the formulae given in section \ref{sec:soln}. Because the compatibility conditions \eqn{cc} from every order of the spectral parameter were used in the construction of the formulae of section \ref{sec:soln}, we can be sure that the entire system of compatibility conditions is consistent and we only need to check one of them (from each matrix row) to find the associated coupled nonlinear system.

Noting \eqn{para}, at $i=k$ the $[1,2]$ entry of the compatibility condition \eqn{cc12} reads
\beqn
\frac{y}{\xs}\hat{a}_k+\frac{\ys}{\xas}=\frac{y}{\xa}+\frac{\ya}{\xas}d_k.\label{cc12e}
\eeqn
Using \eqn{ansatz} and the formula \eqn{solnb}, we find that
\beqna
\psi_k&=&A_k-P_k\varphi_k,\nonumber\\
a_k&=&\frac{x}{\os{2k}{x}}+\frac{x}{\xa}\sum_{j=1}^k\left(\frac{\os{2j}{y}}{\os{2j-1}{y}}+\frac{\os{2j-1}{x}}{\os{2j-2}{x}}\right),\label{ak}\\
d_k&=&\frac{y}{\os{2k}{y}}+\frac{y}{\ya}\sum_{j=1}^k\left(\frac{\os{2j}{x}}{\os{2j-1}{x}}+\frac{\os{2j-1}{y}}{\os{2j-2}{y}}\right).\label{dk}
\ee
Substituting these values into \eqn{cc12e} and into the corresponding compatibility condition from the $[2,1]$ entry \eqn{cc21} gives us the coupled pair of nonlinear equations associated with this Lax pair.
\begin{subequations}\label{lmkdvk}
\beqna
\frac{\xas}{\os{2k+1}{x}}+\frac{\os{2k+1}{y}}{\os{2k}{y}}&=&\frac{\xas}{\xa}+\frac{\ya}{\os{2k}{y}},\\
\frac{\yas}{\os{2k+1}{y}}+\frac{\os{2k+1}{x}}{\os{2k}{x}}&=&\frac{\yas}{\ya}+\frac{\xa}{\os{2k}{x}}.
\ee
\end{subequations}
In the case where $k=0$, \eqn{lmkdvk} is the LMKdV$_2$ equation from \cite{h09}. Indeed, \eqn{lmkdvk} is simply LMKdV$_2$ on a sheared lattice, as depicted in Figure \ref{fig:shear}.
\begin{figure}[h!]
\includegraphics{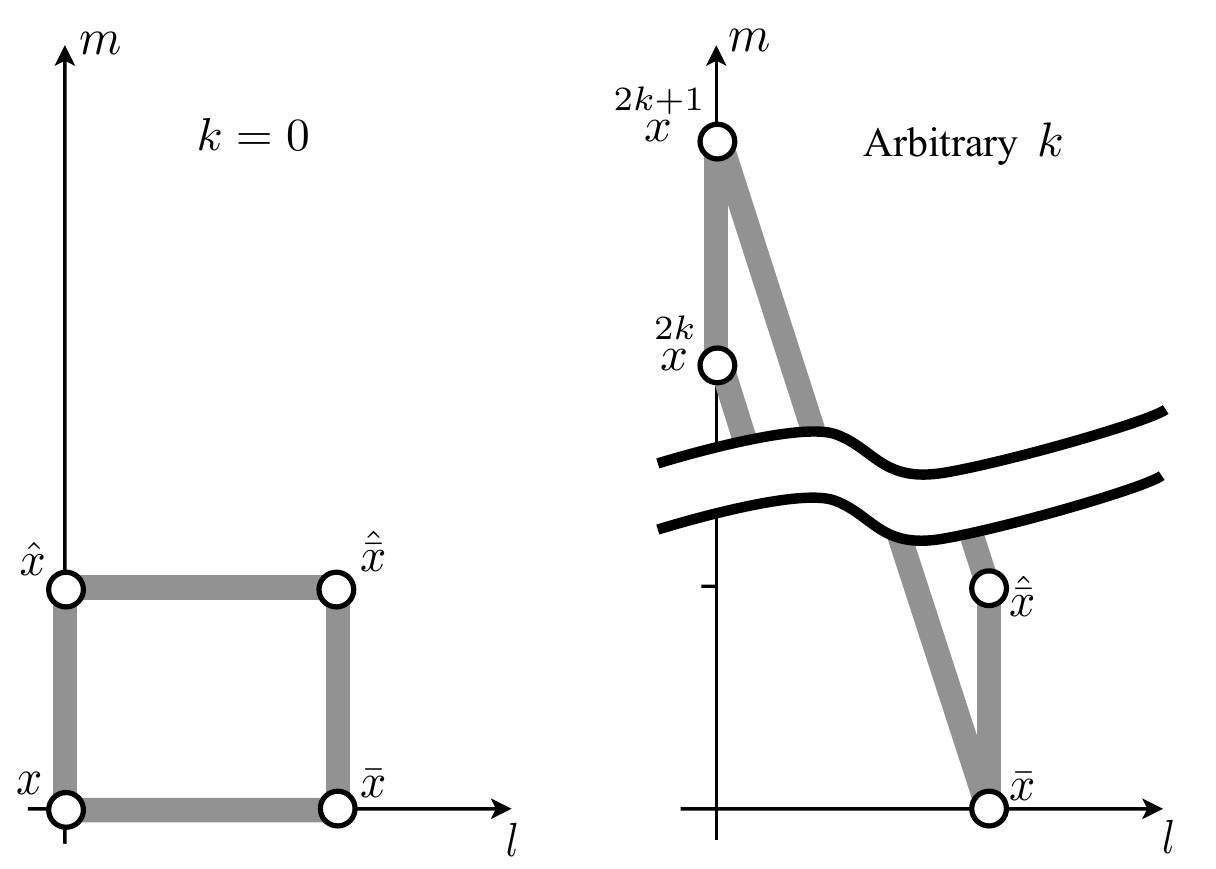}
\caption{Elementary quadralateral of LMKdV$_2$ for $k=0$ (left) and the points occupied on the sheared lattice by a member of the hierarchy with arbitrary $k$ (right).}
\label{fig:shear}
\end{figure}

LMKdV$_2$ is reduced to the familiar LMKdV by setting $y=\la\mu/x$ where $\la(l)$ and $\mu(m)$ are arbitrary functions. We can perform a similar operation in equations \eqn{lmkdvk}, setting $y=\al/x$ where $\al(l,m)$ must satisfy
\beqn
\os{k}{\al}\,\hat{\bar{\al}}=\bar{\al}\,\os{k+1}{\al}\nonumber
\eeqn
so that each of the equations in \eqn{lmkdvk} reduce to the same system. This implies that $\al$ must be of the form
\beqn
\al(l,m)=\la(l)\prod_{i=1}^{k}f_i(kl+m)^{\zeta^{im}},
\eeqn
where $\zeta^k=1$ and $f_i$, $1\leq i\leq k$, as well as $\la(l)$, are arbitrary. Therefore, the reduced equations contain $k+1$ arbitrary functions.

If we change variables such that $l'=l$ and $m'=kl+m$ then $\al'(l',m')=\la'(l')\mu'(m')$, $\la'$ and $\mu'$ arbitrary, and exactly LMKdV is retrieved.

\section{Discussion}
We have constructed a hierarchy of integrable partial difference equations associated with LMKdV. We have obtained explicit formulae that describe all equations in the hierarchy and all of their Lax pairs. We believe this is the first publication of such a result.

It is remarkable that the equations retain the same simple form throughout the hierarchy, even though the Lax pairs become very complicated and contain a large number of terms as the order of the polynomial in the spectral parameter is increased. Indeed, the solutions to the recursion relations from section \ref{sec:soln} contain far more terms than the Lax pairs, all of which file together perfectly to form the equations.

At present the solutions to the recursion relations from section \ref{sec:soln} are somewhat mysterious. The recursion relations are nonlinear and it will be interesting to understand how these solutions relate to the other known solution types.

LMKdV is not a special case, there is evidence to suggest that it is possible to construct similar hierarchies for other lattice systems with $2\times2$ Lax pairs, which, for example, includes all of the equations on the ABS list \cite{abs03}. Further work may also yield similar results in the widely studied continuous and semi-continuous domains.

\section*{acknowledgements}
The author acknowledges support from the Global COE Program “Education and Research Hub for Mathematics-for-Industry” from the Ministry of Education, Culture, Sports, Science and Technology, Japan. This research was also supported in part by the Australian Research Council Discovery grant no. DP110102001.

\appendix
\section{Proof of Theorem \ref{thm:soln}}\label{app:solnpf}
This appendix demonstrates why \eqn{solnp}, \eqn{solnq}, \eqn{solnr} and \eqn{solns} are solutions to \eqn{rr1} and \eqn{rr2}. Because it is the least complicated, we concentrate on the [1,2] entry of \eqn{rr1r}, but the other parts are done in a similar way.

Substituting \eqn{solnp} and \eqn{solnq} into \eqn{rr1r} returns
\beqn
Q_{i+1}^{(j)}=\Ga_i\left(\begin{array}{cc}\dots &\begin{array}{c}\frac{\bea}{\be}\widehat{s_1}^{(j)}_i-\frac{\bea}{\ga}{s_2}^{(j)}_i-\frac{\bea}{(\be\ga)^{j-i}}+\sum_{g=0}^{2i}\{\frac{1}{\be\ga-\os{g+1}{\be\ga}}\times\\
\left[(\os{g+1}{\be\ga})^{1-j}-\left(\frac{\be\ga}{\os{g+1}\be\ga}\right)^i (\be\ga)^{1-j} \right]
\left(\frac{\bea}{\be}\widehat{s_1}^{(g)}_i-\bea\widehat{t_1}^{(g)}_i\right)\}\end{array}\\
&\\
\dots & \begin{array}{c}\frac{\ga}{(\be\ga)^{j-i}}+\sum_{g=0}^{2i}\{\frac{1}{\be\ga-\os{g+1}{\be\ga}}\times\\
\left[(\os{g+1}{\be\ga})^{1-j}-\left(\frac{\be\ga}{\os{g+1}{\be\ga}}\right)^{i}(\be\ga)^{1-j}\right]
\left(-\widehat{s_2}^{(g)}_i+\ga \widehat{t_2}^{(g)}_i\right)\}\end{array}  \end{array}\right),\label{pf1_1}
\eeqn
where we have set ${s_1}^h_i=s^h_i$, while ${s_2}^h_i$ is obtained from ${s_1}^h_i$ by interchanging $\be$ and $\ga$, and $t^h_i$ follows the same pattern. Note that since we will concentrate on the [1,2] element of $Q^{(j)}_{i+1}$, the first column of the matrix on the RHS of \eqn{pf1_1} is irrelevant. To derive the above equation we have used
\beqna
\sum_{h=i+1}^{j-1}S^{(h)}_iD^{(j)}_h&=&\sum_{h=i+1}^{j-1}\frac{1}{(\be\ga)^{j-h}}S^{(h)}_i\left(\begin{array}{cc}-1&\bea/\be\\ \gaa/\ga&-1\end{array}\right),\nonumber\\
&=&\sum_{h=i+1}^{j-1}\frac{1}{(\be\ga)^{j-h}}\sum_{g=0}^{2i}\frac{1}{(\os{g+1}{\be\ga})^{h}}\left(\begin{array}{cc}-{s_1}^{(g)}_i&\frac{\bea}{\be}{s_1}^{(g)}_i\\ \frac{\gaa}{\ga}{s_2}_i^{(g)}&-{s_2}^{(g)}_i\end{array}\right),\nonumber\\
&=&\sum_{g=0}^{2i}\frac{1}{(\os{g+1}{\be\ga})^i}\frac{(\os{g+1}{\be\ga})^{i-j+1}-(\be\ga)^{i-j+1}}{\be\ga-\os{g+1}{(\be\ga)}}\left(\begin{array}{cc}-{s_1}^{(g)}_i&\frac{\bea}{\be}{s_1}^{(g)}_i\\ \frac{\gaa}{\ga}{s_2}_i^{(g)}&-{s_2}^{(g)}_i\end{array}\right)\nonumber
\ee
and similar expressions where necessary.

Expanding \eqn{pf1_1} and considering only the [1,2] element shows that
\beqn
0=\frac{1}{(\os{2i+1}{\be\ga})^j}\Theta_{2i+1}+\sum_{g=1}^{2i}\frac{1}{(\os{g}{\be\ga})^j}\Theta_g+\frac{1}{(\be\ga)^j}\Theta_0,\nonumber
\eeqn
where $\Theta_{2i+1}$, $\Theta_g$, $0<g\leq2i$, and $\Theta_0$, which are given explicitly below, must all be identically zero because of the independence of $j$ in this expression. We therefore have three equations which must be satisfied, these are
\beqn
\Theta_{2i+1}=0=\be\ga \widehat{s_1}^{(2i)}_{i}-\os{2i+1}{\be\ga}\be\widehat{t_1}^{(2i)}_{i}+\frac{\be}{\gas}\prod_{g=1}^i\frac{\os{2g-1}{\be}\,\os{2g}{\ga}}{\os{2g}{\be}\,\os{2g+1}{\ga}}\times\os{2i+1}{\be\ga} (-\widehat{s_2}_{i}^{(2i)}+\ga \widehat{t_2}^{(2i)}_{i}),\label{gao12}
\eeqn
\beqna
\Theta_g=0&=&\ga\widehat{s_1}^{(g-1)}_{i}+\frac{1}{\ga}((\os{g}{\be\ga})-\be\ga){s_2}^{(g)}_i-(\os{g}{\be\ga})\widehat{t_1}^{(g-1)}_{i}\nonumber\\
&&\quad+\frac{1}{\gas}\prod_{h=1}^{i}\frac{\os{2h-1}{\be}\,\os{2h}{\ga}}{\os{2h}{\be}\,\os{2h+1}{\ga}}\times(-\widehat{s_2}^{(g-1)}_{i}+\ga\widehat{t_2}^{(g-1)}_{i}),\label{zhong12}
\ee
and
\beqna
&&\Theta_0=0=\frac{\ga}{\gas}\prod_{h=1}^{i}\frac{\os{2h-1}{\be}\,\os{2h}{\ga}}{\os{2h}{\be}\,\os{2h+1}{\ga}}-\frac{{s_2}^{(0)}_i}{\ga(\be\ga)^i}-1+\label{di12}\\
&&\sum_{g=0}^{2i}\frac{\be\ga}{[\be\ga-(\os{g+1}{\be\ga})](\os{g+1}{\be\ga})^i}\left[\widehat{t_1}^{(g)}_i-\frac{1}{\be}\widehat{s_1}^{(g)}_i+\left(\prod_{h=1}^{i}\frac{\os{2h-1}{\be}\,\os{2h}{\ga}}{\os{2h}{\be}\,\os{2h+1}{\ga}}\right)\left(\frac{1}{\gas}\widehat{s_2}^{(g)}_i-\frac{\ga}{\gas}\widehat{t_2}^{(g)}_i\right)\right].\nonumber
\ee

First we show that \eqn{gao12} is satisfied. Using \eqn{solnp} and \eqn{solnq}, as well as \eqn{sumdelta} and \eqn{xi} we obtain
\beqna
0=-\ga\sum_{h_1=1}^{2i+1}\ldots\sum_{h_{2i}=1}^{h_{2i-1}-1+\delta_{2i+1}^{h_{2i-1}}}\os{h_1}{\ga}\os{h_2}{\be}\ldots\os{h_{2i-1}}{\ga}\os{h_{2i}}{\be}-\sum_{h_1=1}^{2i+1}\ldots\sum_{h_{2i+1}=0}^{h_{2i}-1+\delta_{2i+1}^{h_{2i}}}\os{h_1}{\ga}\os{h_2}{\be}\ldots\os{h_{2i}}{\be}\,\os{h_{2i+1}}{\ga}+\nonumber\\
\os{2i+1}{\ga}\sum_{h_1=1}^{2i+1}\ldots\sum_{h_{2i}=1}^{h_{2i-1}-1+\delta_{2i+1}^{h_{2i-1}}}\os{h_1}{\be}\os{h_2}{\ga}\ldots\os{h_{2i-1}}{\be}\os{h_{2i}}{\ga}+\frac{\ga}{\os{2i+1}{\be}}\sum_{h_1=1}^{2i+1}\ldots\sum_{h_{2i+1}=1}^{h_{2i}-1+\delta_{2i+1}^{h_{2i}}}\os{h_1}{\be}\os{h_2}{\ga}\ldots\os{h_{2i}}{\ga}\os{h_{2i+1}}{\be}.\nonumber
\ee
However, the sums arising here can be re-written using \eqn{qitop}, \eqn{qibot} and the arguments used to derive those. Therefore, we may write the second sum from above as
\beqn
\sum_{h_1=1}^{2i+1}\ldots\sum_{h_{2i+1}=1}^{h_{2i}-1+\delta_{2i+1}^{h_{2i}}}\os{h_1}{\ga}\ldots\os{h_{2i+1}}{\ga}=\os{2i+1}{\ga}\sum_{h_1=1}^{2i+1}\ldots\sum_{h_{2i}=1}^{h_{2i-1}-1+\delta_{2i+1}^{h_{2i-1}}}\os{h_1}{\be}\ldots\os{h_{2i}}{\ga}\nonumber
\eeqn
and the fourth sum is similar. This shows that \eqn{gao12} is satisfied.

Turning to \eqn{zhong12}, we agian substitute \eqn{solnp} and \eqn{solnq}, use \eqn{sumdelta} and \eqn{xi}, then cancel common terms to arrive at
\beqna
0&=&-\os{2i+1}{\be}\,\ga\sum_{h_1=1}^{2i+1}\ldots\sum_{h_{2i}=1}^{h_{2i-1}-1+\delta_{g}^{h_{2i-1}}}\os{h_1}{\ga}\ldots\os{h_{2i}}{\be} +\ga\sum_{h_1=1}^{2i+1}\ldots\sum_{h_{2i+1}=1}^{h_{2i}-1+\delta_{g}^{h_{2i}}}\os{h_1}{\be}\ldots\os{h_{2i+1}}{\be}  \nonumber\\
&&-\os{2i+1}{\be}\sum_{h_1=1}^{2i+1}\ldots\sum_{h_{2i+1}=1}^{h_{2i}-1+\delta_{g}^{h_{2i}}}\os{h_1}{\ga}\ldots\os{h_{2i+1}}{\ga}+(\os{g}{\be\ga})\sum_{h_1=1}^{2i+1}\ldots\sum_{h_{2i}=1}^{h_{2i-1}-1+\delta_{g}^{h_{2i-1}}}\os{h_1}{\be}\ldots\os{h_{2i}}{\ga}\nonumber\\
&&+ [(\os{2i+1}{\be\ga})-(\os{g}{\be\ga})]\sum_{h_1=0}^{2i}\ldots\sum_{h_{2i}=0}^{h_{2i-1}-1+\delta_{g}^{h_{2i-1}}}\os{h_1}{\be}\ldots\os{h_{2i}}{\ga}.\nonumber
\ee
It is easy to see that this cancels by using \eqn{qitop}, \eqn{qibot} and \eqn{qibg}. Thus, \eqn{zhong12} is satisfied. We also note that  equations \eqn{sumdelta} and \eqn{xi} are not essential to show that \eqn{gao12} and \eqn{zhong12} are satisfied. Although it is more complicated, we can use the solutions \eqn{solnp} and \eqn{solnq} directly. Doing so leads to a cancelation of terms similar to that described in Figure \ref{fig:termsmap}, see \eqn{etadecomp}, \eqn{xidecomp1} and \eqn{xidecomp2}, below, and the surrounding arguments for an explanation of this figure.

Finally, consider \eqn{di12}. Using the \eqn{solnp} and \eqn{solnq} we can rewrite \eqn{di12} as
\beqna
0=\ga\hat{\eta}^{(0)}_{2i}(\beta)-\os{2i+1}{\be}\hat{\eta}^{(0)}_{2i}(\ga)+\frac{\widehat{(\be\ga)}\prod_{j=1}^{i}\os{2j}{(\be\ga)}\os{2j+1}{(\be\ga)}}{\prod_{j=1}^{2i}[\os{j}{(\be\ga)}-\be\ga]}\sum_{h=0}^i(\be\ga)^h\xi_{2i}^{(2h)}(\be)\nonumber\\
+\widehat{(\be\ga)} \prod_{j=1}^i\os{2j}{(\be\ga)}\os{2j+1}{(\be\ga)}\cdot\sum_{g=0}^{2i}\left[\frac{\be\ga/\os{g+1}{(\be\ga)}}{\prod_{j=0, j\neq g+1}^{2i+1}[\os{j}{(\be\ga)}-\os{g+1}{(\be\ga)}]}\sum_{h=0}^i[\os{g+1}{(\be\ga)}]^h[\os{2i+1}{\be}\,\hat{\eta}^{(2h)}_{2i}(\ga)\right.\nonumber\\
\left.+\frac{\os{2i+1}{\be}}{\be}\os{g+1}{(\be\ga)}\hat{\xi}^{(2h)}_{2i}(\ga)-\os{g+1}{(\be\ga)}\hat{\xi}^{(2h)}_{2i}(\be)-\ga\hat{\eta}^{(2h)}_{2i}(\be)] \right] .\quad\quad\label{xia12}
\ee

In the following, we often need to differentiate between terms that contain multiple factors shifted by the same amount in $m$ and terms that do not. Let us refer to any $\be$ and $\ga$ shifted by the same amount in $m$, {\em i.e.} $\os{g}{(\be\ga)}$, as `repetitively shifted'. Let us refer to those terms that do not contain any $\be$ and $\ga$ shifted by the same amount in $m$ as `cascading'. Both repetitively shifted and cascading terms can be seen in Figure \ref{fig:sumDiagram}. In that figure, terms containing repetitive shifts must all have at least two nodes along the horizontal line through $g$, while cascading terms have at most one node along this line. 

To show that \eqn{xia12} holds, we factor out all of the repetitively shifted terms and focus on the cascading terms. The strategy is to show that the cascading terms group together in certain ways and that the repetitively shifted terms multiplying these groups all cancel, proving that \eqn{xia12} is satisfied.

In \eqn{xia12}, the quantities that contain only cascading terms are: 
\beqn
\xi_{2i}^{(2h)}(\be), \quad\os{2i+1}{\be}\,\hat{\eta}^{(2h)}_{2i}(\ga),\quad\frac{\os{2i+1}{\be}}{\be}\hat{\xi}^{(2h)}_{2i}(\ga),\quad\hat{\xi}^{(2h)}_{2i}(\be),\quad\ga\hat{\eta}^{(2h)}_{2i}(\be),\quad0\leq h\leq 2i
\eeqn
 (note that some terms within $\frac{\os{2i+1}{\be}}{\be}\hat{\xi}^{(2h)}_{2i}(\ga)$ contain a factor of $\os{2i+1}{(\be\ga)}$, which is taken into consideration below).

Obviously we may use the definition \eqn{eta} and write
\beqna
\eta^{(2h)}_{2i}(\ga)&=&\sum_{h_1=0}^{2i}\sum_{h_2=0}^{h_1-1}\ldots\sum_{h_{2(i-h)}=0}^{h_{2(i-h)-1}-1}\sum_{h_{2(i-h)+1}=0}^{h_{2(i-h)}-1}\os{h_1}{\ga}\,\os{h_2}{\be}\ldots\os{h_{2(i-h)}}{\be}\,\os{h_{2(i-h)+1}}{\ga},\nonumber\\
&=&\os{2i}{\ga}\sum_{h_2=0}^{2i-1}\ldots\sum_{h_{2(i-h)+1}=0}^{h_{2(i-h)}-1}\os{h_2}{\be}\ldots\os{h_{2(i-h)}}{\be}\,\os{h_{2(i-h)+1}}{\ga}\nonumber\\
&&\quad\quad+\sum_{h_1=0}^{2i-1}\ldots\sum_{h_{2(i-h)+1}=0}^{h_{2(i-h)}-1}\os{h_1}{\ga}\,\os{h_2}{\be}\ldots\os{h_{2(i-h)}}{\be}\,\os{h_{2(i-h)+1}}{\ga},\nonumber\\
&=&\eta^{(2h)}_{2i}(\ga)\bigg|_{h_1=2i}+\eta^{(2h)}_{2i}(\ga)\bigg|_{h_1< 2i}, \quad 0\leq h\leq2i.\label{etadecomp}
\ee
Notice that when $h=0$, $\eta_{2i}^{(0)}$ has $2i+1$ sums, each with a decreasing upper limit. Thus, $\eta^{(0)}_{2i}$ must have $h_1=2i$, so the second term in \eqn{etadecomp} is zero when $h=0$. The other terms in \eqn{xia12} involving $\eta$ can also be decomposed in a similar manner.

We may decompose the terms involving $\xi$ into parts that are related to $\eta$ as follows.
\beqna
\xi^{(2h)}_{2i}(\be)&=&\sum_{h_1=0}^{2i}\sum_{h_2=0}^{h_1-1}\ldots\sum_{h_{2(i-h)}=0}^{h_{2(i-h)-1}-1}\os{h_1}{\be}\,\os{h_2}{\ga}\ldots\os{h_{2(i-h)-1}}{\be}\,\os{h_{2(i-h)}}{\ga},\nonumber\\
&=&\os{2i}{\be}\sum_{h_2=0}^{2i-1}\ldots\sum_{h_{2(i-h)}=0}^{h_{2(i-h)-1}-1}\os{h_2}{\ga}\ldots\os{h_{2(i-h)-1}}{\be}\,\os{h_{2(i-h)}}{\ga}\nonumber\\
&&\quad\quad+\sum_{h_1=0}^{2i-1}\ldots \sum_{h_{2(i-h)}=0}^{h_{2(i-h)-1}-1}\os{h_1}{\be}\,\os{h_2}{\ga}\ldots\os{h_{2(i-h)-1}}{\be}\,\os{h_{2(i-h)}}{\ga},\nonumber\\
&=&\os{2i}{\be}\sum_{h_1=0}^{2i-1}\ldots\sum_{h_{2(i-h-1)+1}=0}^{h_{2(i-h-1)}-1}\os{h_1}{\ga}\ldots\os{h_{2(i-h-1)}}{\be}\,\os{h_{2(i-h-1)+1}}{\ga}\nonumber\\
&&\quad\quad+\frac{\os{2i}{\be}}{\os{2i}{(\be\ga)}}\os{2i}{\ga}\sum_{h_2=0}^{2i-1}\ldots \sum_{h_{2(i-h)+1}=0}^{h_{2(i-h)}-1}\os{h_2}{\be}\,\os{h_3}{\ga}\ldots\os{h_{2(i-h)}}{\be}\,\os{h_{2(i-h)+1}}{\ga},\nonumber\\
\xi^{(2h)}_{2i}(\be)&=&\os{2i}{\be}\eta^{(2h+2)}_{2i}(\ga)\bigg|_{h_1< 2i}+\frac{\os{2i}{\be}}{\os{2i}{(\be\ga)}}\eta^{(2h)}_{2i}(\ga)\bigg|_{h_1=2i}.\label{xidecomp1}
\ee
We may also decompose $\xi$ in a different way as follows.
\beqna
\xi^{(2h)}_{2i}(\be)&=&\sum_{h_1=0}^{2i}\ldots\sum_{h_{2(i-h)-1}=0}^{h_{2(i-h)-2}-1}\sum_{h_{2(i-h)}=0}^{h_{2(i-h)-1}-1}\os{h_1}{\be}\,\os{h_2}{\ga}\ldots\os{h_{2(i-h)-1}}{\be}\,\os{h_{2(i-h)}}{\ga},\nonumber\\
&=&\ga\sum_{h_1=1}^{2i}\ldots\sum_{h_{2(i-h)-1}=1}^{h_{2(i-h)-2}-1}\os{h_1}{\be}\os{h_2}{\ga}\ldots\os{h_{2(i-h)-1}}{\be}\nonumber\\
&&\quad\quad+\sum_{h_1=1}^{2i}\ldots \sum_{h_{2(i-h)}=1}^{h_{2(i-h)-1}-1}\os{h_1}{\be}\,\os{h_2}{\ga}\ldots\os{h_{2(i-h)-1}}{\be}\,\os{h_{2(i-h)}}{\ga},\nonumber\\
&=&\ga\sum_{h_1=0}^{2i-1}\ldots\sum_{h_{2(i-h-1)+1}=0}^{h_{2(i-h-1)}-1}\hat{\os{h_1}{\be}}\hat{\os{h_2}{\ga}}\ldots\widehat{\os{h_{2(i-h-1)+1}}{\be}}\nonumber\\
&&\quad\quad+\frac{\os{2i}{\be}}{\os{2i}{(\be\ga)}}\os{2i}{\ga}\sum_{h_2=0}^{2i-1}\ldots \sum_{h_{2(i-h)+1}=0}^{h_{2(i-h)}-1}\hat{\os{h_2}{\be}}\,\hat{\os{h_3}{\ga}}\ldots\widehat{\os{h_{2(i-h)}}{\be}}\,\widehat{\os{h_{2(i-h)+1}}{\ga}},\nonumber\\
\xi^{(2h)}_{2i}(\be)&=&\ga\hat{\eta}^{(2h+2)}_{2i}(\be)\bigg|_{h_1< 2i}+\frac{\os{2i}{\be}}{\os{2i}{(\be\ga)}}\hat{\eta}^{(2h)}_{2i}(\ga)\bigg|_{h_1=2i}.\label{xidecomp2}
\ee
Using \eqn{etadecomp}, \eqn{xidecomp1} and \eqn{xidecomp2} it is easy to show that the cascading terms group together as shown in Figure \ref{fig:termsmap}, note that any repetitively shifted factors are omitted from that figure. We give directions about how to use this figure below.

Before we do that, let us recall two so-called `Euler Identities' which are crucial to what follows. For an arbitrary set of $i+1$ distinct points $a_j\in\mathbb{C}$, $0\leq j\leq i$, the first identity is:
\beqn
\prod_{j=0}^{i}a_j\cdot\sum_{g=0}^{i}\frac{1}{a_g\prod_{h=0,\,h\neq g}^{i}(a_h-a_g)}=1,\label{id1}
\eeqn
and the second identity is:
\beqn
\sum_{g=0}^{i}\frac{a_g^p}{\prod_{h=0,\,h\neq g}^{i}(a_h-a_g)}=1,\label{id2}
\eeqn
for an integer $p$ with $0\leq p<i$. These two identities are either given explicitly in, or easily derived from arguments outlined in, section 2.7 of \cite{h98}.

\begin{figure}[h!]
\includegraphics{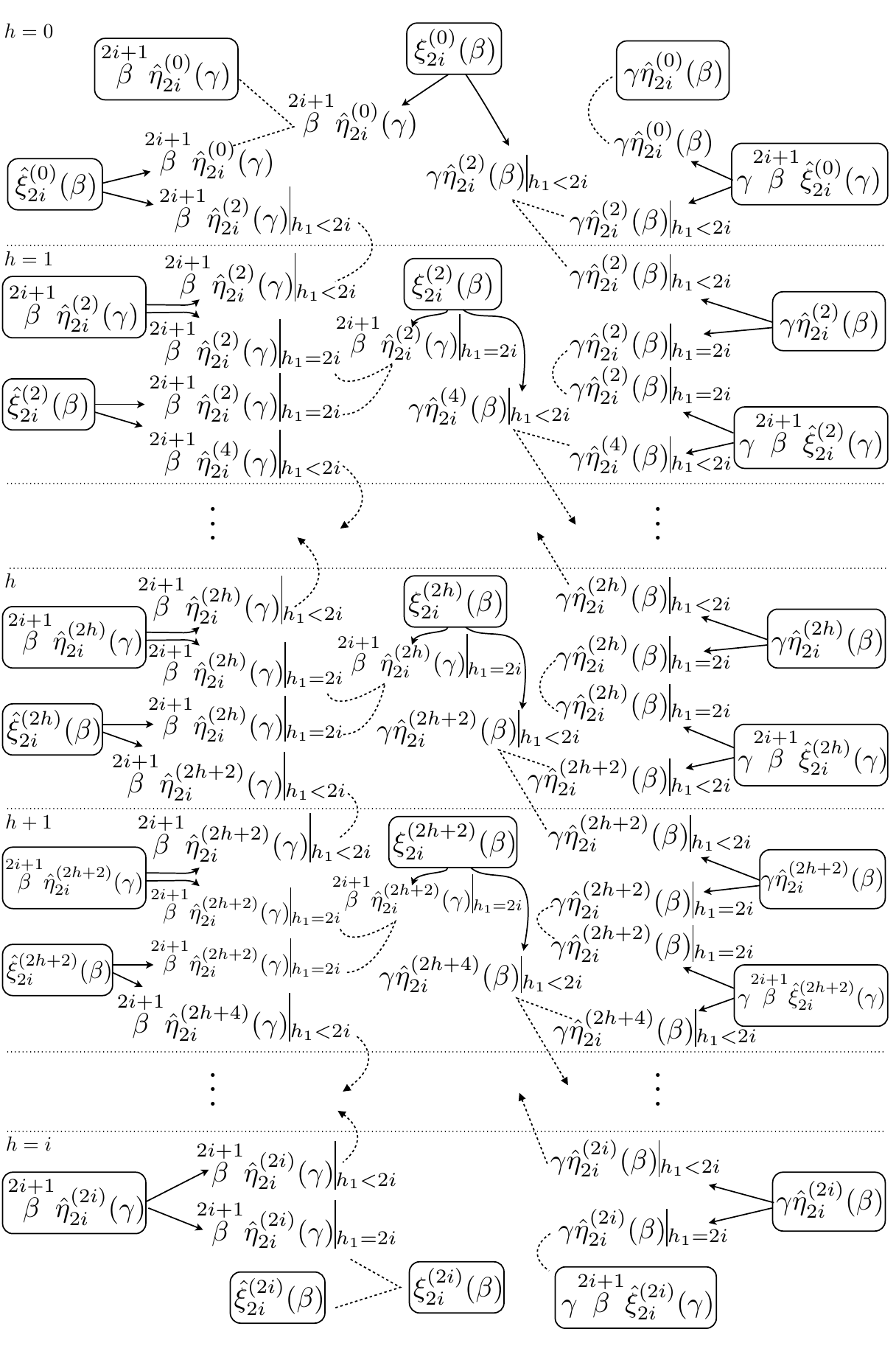}
\caption{Relationship between the various terms that appear in \eqn{xia12}. Terms that appear explicitly in the equation are boxed. Solid arrows indicate decomposition of a sum into two terms. Dashed lines link groups of like terms. The variable $h$ is the same as that from the inner sum of the last term in \eqn{xia12}.}
\label{fig:termsmap}
\end{figure}

Next we turn to the coefficients of each of the groups of cascading terms, all such coefficients are made up of repetitively shifted terms. Consider the top level of Figure \ref{fig:termsmap}, which corresponds to $h=0$. This $h$ also corresponds to the index from the sum in the third term, and the inner sum in the last term, of \eqn{xia12}. Due to \eqn{xidecomp1} and \eqn{xidecomp2} we find that $\hat{\xi}_{2i}^{(0)}(\be)$, $\os{2i+1}{\be}\hat{\eta}_{2i}^{(0)}(\ga)$ and $\xi_{2i}^{(0)}(\be)$ have common terms in their decompositions, as depicted by a dotted line in Figure \ref{fig:termsmap}. The coefficients multiplying these terms in \eqn{xia12} sum to:
\beqna
0&=&-1+\frac{\widehat{(\be\ga)}}{\os{2i+1}{(\be\ga)}}\frac{\prod_{h=1}^i\os{2h}{(\be\ga)}\os{2h+1}{(\be\ga)}}{\prod_{h=1}^{2i}[\os{h}{(\be\ga)}-\be\ga]}\nonumber\\
&&\quad\quad+\widehat{(\be\ga)}\prod_{h=1}^i\os{2h}{(\be\ga)}\os{2h+1}{(\be\ga)}\cdot\sum_{g=0}^{2i}\frac{\be\ga/\os{g+1}{(\be\ga)}[1-\os{g+1}{(\be\ga)}/\os{2i+1}{(\be\ga)}]}{\prod_{h=0, h\neq g+1}^{2i+1}[\os{h}{(\be\ga)}-\os{2g+1}{(\be\ga)}]},\nonumber\\
\Rightarrow \quad1&=&\prod_{h=0}^{2i}\os{h}{(\be\ga)}\cdot \sum_{g=0}^{2i}\frac{1}{\os{g}{(\be\ga)}}\frac{1}{\prod_{h=0,h\neq g}^{2i}[\os{h}{(\be\ga)}-\os{g}{(\be\ga)}]}.
\ee
This is the same as \eqn{id1}.

The other parts of \eqn{xia12}, at the level of $h=0$ in Figure \ref{fig:termsmap} are cancelled in a similar way. We turn our attention to the arbitrary level $0<h<2i$ of that figure. As an example, on the right hand side we observe the following set of terms with common factors:
\beqn
\ga \os{2i+1}{\be}\hat{\xi}^{(2h)}_{2i}(\ga), \quad \ga\hat{\eta}^{(2h+2)}_{2i}(\be)\bigg|_{h_1<2i}, \quad \xi^{(2h)}_{2i}(\be).\nonumber
\eeqn
Adding together the relevant coefficients from \eqn{xia12} leads to
\beqna
0&=&\widehat{(\be\ga)}\frac{\prod_{j=1}^i\os{2j}{(\be\ga)}\os{2j+1}{(\be\ga)}}{\prod_{j=1}^{2i}[\os{j}{(\be\ga)}-\be\ga]}(\be\ga)^h\nonumber\\
&&+\widehat{(\be\ga)}\prod_{j=1}^i\os{2j}{(\be\ga)}\os{2j+1}{(\be\ga)}\sum_{g=1}^{2i}\frac{\be\ga/\os{g+1}{(\be\ga)}}{\prod_{j=0,j\neq g}^{2i+1}[\os{j}{(\be\ga)}-\os{g+1}{(\be\ga)}]}[\os{g+1}{(\be\ga)}]^{h+1}[\os{2i+1}{(\be\ga)}/(\be\ga)-1]\nonumber\\
0&=&\sum_{g=0}^{2i+1}\frac{[\os{g}{(\be\ga)}]^h}{\prod_{j=0,j\neq g}^{2i+1}[\os{h}{(\be\ga)}-\os{g}{(\be\ga)}]}, \quad 0<h<2i.\nonumber
\ee
This is the same as \eqn{id2}. One can check in a similar way that all of the parts of \eqn{xia12} are either identically satisfied or equivalent to one of \eqn{id1} or \eqn{id2}. We point out that the solutions match the truncation of the sums in \eqn{xia12} because $\xi_{2i}^{(2i)}=1$.

Thus, the $(1,2)$ entry of \eqn{rr1r} is satisfied by the solutions \eqn{solnp} and \eqn{solnq}. In a similar way, it can be shown that each of the recursion relations in \eqn{rr1} and \eqn{rr2} are also satisfied by \eqn{solnp}, \eqn{solnq}, \eqn{solnr} and \eqn{solns}.


\end{document}